\documentclass[12pt,reqno]{amsart}

\usepackage{color}

\usepackage{fullpage}

\usepackage{epsfig}
\usepackage{amsmath, amsthm, amsfonts}
\usepackage{graphicx}

\numberwithin{equation}{section}

\newcommand{\R}{{\mathbb R}}
\newcommand{\C}{{\mathbb C}}

\newcommand{\T}{{\mathbf T}}
\renewcommand{\d}{\partial}

\newcommand{\Ai}{{\operatorname{Ai}}}

\newcommand{\ep}{\varepsilon}
\newcommand{\de}{\delta}

\newcommand{\z}{\zeta}
\newcommand{\Th}{\Theta}

\newcommand{\B}{{\mathbf B}}
\newcommand{\A}{{\mathbf A}}
\newcommand{\1}{{\mathbf 1}}
\newcommand{\D}{{\mathbf D}}
\newcommand{\V}{{\mathbf V}}
\newcommand{\U}{{\mathbf U}}

\newtheorem{theo}{{\sc \bf Theorem}}[section]

\newtheorem{lem}[theo]{{\sc \bf Lemma}}
\newtheorem{prop}[theo]{{\sc \bf Proposition}}

\newtheorem{lemma}{Lemma}[section]

\newcommand{\RR}{R}
\newcommand{\id}{\1}

\newcommand{\bfR}{{\bf R}}
\newcommand{\pd}{\partial}

\newcommand{\M}{{\mathbf M}}

\newcommand{\calM}{{\mathcal M}}
\newcommand{\calT}{{\mathcal T}}
\newcommand{\Qua}{\psi}
\newcommand{\QQua}{\Psi}
\newcommand{\K}{{\mathbf K}}

\begin{document}


\title{On the joint distribution of the maximum and its position of the Airy$_2$ process minus a parabola} 



\author{Jinho Baik}
\address{Department of Mathematics,
University of Michigan,
530 Church St., Ann Arbor, MI 48109, U.S.A.}
\email{baik@umich.edu}

\author{Karl Liechty}
\address{Department of Mathematics,
University of Michigan,
530 Church St., Ann Arbor, MI 48109, U.S.A.}
\email{kliechty@umich.edu}

\author{Gr\'egory Schehr}
\address{Laboratoire de Physique Th\'eorique et Mod\`eles Statistiques, Universit\'e Paris Sud 11 and CNRS.}
\email{gregory.schehr@u-psud.fr}

%


\date{\today}

\begin{abstract}
The maximal point of the Airy$_2$ process minus a parabola is believed to describe  
the scaling limit of the end-point of the directed polymer in a random medium.
This was proved to be true for a few specific cases. 
Recently two different formulas for the joint distribution of the location and the height of this maximal point
were obtained, one by Moreno Flores, Quastel and Remenik, and the other by Schehr.
The first formula is given in terms of the Airy function and an associated operator, and the second formula is expressed in terms of the Lax pair equations of the Painlev\'e II equation. We give a direct proof that these two formulas are the same.
\end{abstract}

\maketitle


\section{Introduction and result}

Let ${\mathcal A}_2(u)$ be the Airy$_2$ process \cite{PrahoferSpohn}. 
It is a stationary process whose marginal distribution is the Gaussian unitary ensemble (GUE) Tracy-Widom distribution \cite{TW}. 
Let ${\mathcal M}$ and ${\mathcal T}$ be the random variables defined by
\begin{equation}\label{in:2}
	{\mathcal M} := \max_{u \in \mathbb{R}} \left( {\mathcal A}_2(u) - u^2 \right) \;,
\end{equation}
and
\begin{equation}\label{in:3}
	{\mathcal T} := \arg \max_{u \in \mathbb{R}} \left( {\mathcal A}_2(u) - u^2 \right) \;.
\end{equation}
The Airy$_2$ process is widely expected to be the universal limiting process for the spatial fluctuations of the models in the Kardar-Parisi-Zhang (KPZ) universality class \cite{kpz} in $1+1$ dimensions. 
This has been proved for a few specific cases. For example, it is proven in the polynuclear growth model \cite{PrahoferSpohn}
and the directed last passage percolation with geometric or exponential weights \cite{Johansson_dprm}. The Airy$_2$ process also appears in the scaling limit of TASEP,  
random tiling problems, and 1D non-intersecting processes (see for example \cite{Johansson_dprm, Okounkov, Johansson_tile, TWBE, Borodin1, Borodin2, Petrov} and references therein).

In the directed last passage percolations with geometric weights, 
the random variables $\calM$ and $\calT$ represent the following observables. 
Consider a point-to-line percolation. 
It was shown \cite{Johansson_dprm} that, assuming that the $\mathcal{A}_2(u)-u^2$ has a unique maximum almost surely, 
$\calT$ describes the fluctuations of the location of the end-point of the maximizing path.  
The uniqueness assumption was recently proved in \cite{CorwinHammond}. 
On the other hand, $\calM$ is the random variable for the limiting fluctuations of the energy of the maximizing path. 

The marginal distribution of $\calM$ is known to be equal to 
the Gaussian orthogonal ensemble (GOE) Tracy-Widom distribution \cite{TW2}. 
This was proved in \cite{Johansson_dprm}  indirectly using a correspondence\footnote{This correspondence was also previously observed in \cite{Krug}.} between certain observables in the directed last passage percolation in point to point geometry and associated ones in point to line geometry for which the limiting distribution had been evaluated \cite{BR}.
A more direct proof was later obtained \cite{ForresterMajumdarSchehr} by showing that the scaled 
fluctuations of the maximal height of $N$ non-intersecting Brownian excursions is governed by the GOE Tracy-Widom distribution
in the asymptotic large $N$ limit.
The physical asymptotic analysis used in that paper was subsequently proved rigorously in \cite{Liechty} using Riemann-Hilbert techniques. 
Another direct proof based on determinantal process was obtained in \cite{CQR}
using the explicit determinantal formula of the Airy$_2$ process.

Because of the broad relevance of the KPZ universality class, the
distributions of $\calM$ and  $\calT$ have generated some recent interest
in the theoretical literature  \cite{Ferrari, SchehrPRL,KatoriPRE,FeierlJPHYSA,Rambeau1,Rambeau2,LeDoussalCalabreseShort,LeDoussalCalabreseLong,Quastel_tail} as well as in the experimental literature \cite{Kazz1, Kazz2}.    

\bigskip

Exact expressions for the joint distribution of $(\calM, \calT)$ were obtained   in two recent papers: in
\cite{quastel_jpdf} by Moreno Flores, Quastel, and Remenik, 
and  in \cite{schehr_12} by Schehr.
The paper \cite{quastel_jpdf} is mathematical and rigorous, and the formula involves the Airy function and the resolvent of an associated operator. 
On the other hand, the paper \cite{schehr_12} is physical 
and the distribution is expressed in terms of the Lax pair for the Painlev\'e II equation. 
The purpose of this paper is to verify directly that these two expressions are indeed the same. 
In doing so, we describe an explicit solution to the Lax pair equation for the Painlev\'e II equation
(see Proposition~\ref{lax_pair} below). 
This calculation is closely related to the seminal work of Tracy and Widom \cite{TW} on the identification of the Fredholm determinant of the Airy operator in terms of the Painlev\'e II equation.

We now describe the formulas of \cite{quastel_jpdf} and \cite{schehr_12}. 
Let $\hat {P}(m,t)$ denote the joint density function of $(\calM, \calT)$. 
We first describe the formula of \cite{quastel_jpdf}. 
Let $\Ai$ be the Airy function \cite{abramowitz}, and let $\B_s$ be the integral operator acting on $L^2[0,\infty)$ with kernel
\begin{equation}\label{in:4}
	\B_s(x,y)=\Ai(x+y+s)
\end{equation}
for each $s\in \R$. It is known that $\1-\B_s$ is invertible. We set 
\begin{equation}\label{in:5}
	\rho_s(x,y)=(\1-\B_s)^{-1} (x,y),  \qquad x,y\ge 0 \,.
\end{equation}
Define, for $t, m\in \R$, 
\begin{equation}\label{in:6}
	\Qua(x; t, m) = 2 e^{x\, t} [t {\rm Ai}(t^2+m+x) + {\rm Ai}'(t^2+m+x)] \;. \;
\end{equation}
Then the formula of \cite{quastel_jpdf} is 
\begin{equation}\label{in:7}
	\hat {P}(m,t) 
	= 2^{1/3} {\mathcal F}_1(2^{2/3} m) \int_{0}^\infty dx_1 \int_{0}^\infty dx_2 \, \Qua(2^{1/3}x_1;-t,m) \rho_{2^{2/3}m}(x_1,x_2) \, \Qua(2^{1/3} x_2;t,m) \;,
\end{equation}
where ${\mathcal F}_1(s)= \det(\1-\B_s)$ denotes the GOE Tracy-Widom distribution function, see \cite{TW2, FS}.

The formula of \cite{schehr_12} is as follows. 
Let $q(s)$ be the particular solution of the Painlev\'{e} II equation
\begin{equation}\label{in:10}
	q''(s)=2q(s)^3+sq(s)\,,
\end{equation}
satisfying 
\begin{equation}\label{in:11}
	q(s) \sim \Ai(s)\,, \quad \textrm{as} \ s \to + \infty \;.
\end{equation}
This particular solution is known as the Hastings-McLeod solution \cite{HastingsMcLeod, Itsbook},  
and the uniqueness and the global existence are well established. 
Now consider the following Lax pair equations\footnote{This Lax pair is a simple transformation of the Lax pair equations of Flaschka and Newell \cite{FN_80}. If we call the solution of the Flaschka-Newell Lax pair $\hat{\Phi}$, then we have $\hat{\Phi}=\begin{pmatrix} 1 & i \\ 1& -i \end{pmatrix} \Phi$, where $\Phi$ solves (\ref{in:8}). See~\cite{BI03}.}
associated  
to the Hastings-McLeod solution of the Painlev\'e II equation, {\it i.e.} the following linear differential equations for a two-dimensional vector 
$\Phi=\Phi(\zeta,s)$, 
\begin{equation}\label{in:8}
 \frac{\partial}{\partial \zeta} \Phi = A \Phi \;, \quad \; \frac{\partial}{\partial s} \Phi = B \Phi \;,
 \end{equation} 
 where the $2 \times 2$ matrices $A = A(\zeta, s)$ and $B = B(\zeta,s)$ are given by
 \begin{equation}\label{in:9}
 A(\zeta,s) = \left( 
 \begin{array}{c c}
 4 \zeta q &  4 \zeta^2 + s + 2q^2 + 2q'\\
 -4 \zeta^2 - s - 2 q^2 + 2q' & -4 \zeta q
 \end{array}\right) \;
\end{equation}
and 
\begin{equation}\label{in:8-1}
  B(\zeta,s) = \left( 
 \begin{array}{c c}
 q &  \zeta \\
 -\zeta & - q
 \end{array}\right) \;.
\end{equation}
The above system is overdetermined, and the compatibility of the equations implies that $q(s)$ solves the Painlev\'e II equation. 
Now let $\Phi =\begin{pmatrix} \Phi_1 \\ \Phi_2 \end{pmatrix}$ 
be the unique solution of (\ref{in:8}) satisfying 
the real asymptotics
\begin{equation}\label{in:13}
	\Phi_1(\z;s)=\cos\left(\frac{4}{3}\z^3+s\z\right)+O(\z^{-1}), \quad \Phi_2(\z;s)=-\sin\left(\frac{4}{3}\z^3+s\z\right)+O(\z^{-1})\,,
\end{equation}
as $\z \to \pm \infty$ for $s\in \R$. 
There is such a solution (see \cite{BI03},  \cite{DZ}, or \cite{Itsbook}), and it  further satisfies the properties that $\Phi_1(\z;s)$ and $\Phi_2(\z;s)$ are real for real $\z$ and $s$, and 
\begin{equation}\label{in:12}
	\Phi_1(-\z;s)=\Phi_1(\z;s), \quad \Phi_2(-\z;s)=-\Phi_2(\z;s).
\end{equation}
Define
\begin{equation}\label{in:15}
	h(s,w) :=  \int_{0}^\infty  \zeta \Phi_2(\zeta,s) e^{- w \zeta^2} \, d\zeta\,.
\end{equation}
The formula of \cite{schehr_12} is that 
\begin{equation}\label{in:14}
\begin{aligned}
	{\hat P}(m,t) &= 4 P(2^{2/3}m, 2^{4/3} t) \;, 
\end{aligned}
\end{equation}
where
\begin{equation}\label{in:14-1}
\begin{aligned}
	{P}(s,w) &:= \frac{4}{\pi^2} {\mathcal F}_1(s) \int_{s}^\infty h(u,w) h(u,-w) \, du \;.
\end{aligned}
\end{equation}

The main result of this paper is 
\begin{theo}\label{thm}
The two formulas~\eqref{in:7} and~\eqref{in:14} of ${\hat P}(m,t)$ are the same. 
\end{theo}

We give a direct verification of this statement. 
The derivation of the formula (\ref{in:14}) in the work \cite{schehr_12} relies on an Ansatz which is not rigorously justified. 
The above theorem gives an indirect proof of the work of \cite{schehr_12} and puts the formula (\ref{in:14}) on a rigorous mathematical footing.

\subsubsection*{Acknowledgements}
We would like to thank Percy Deift for showing us a simple proof of Lemma \ref{lem1}.
The work of Baik was supported in part by NSF grants DMS-1068646 and the work of Schehr was supported by ANR grant 2011-BS04-013-01 WALKMAT.

\section{Proof of theorem}\label{proof}

We start with the formula (\ref{in:14}) and verify that it is equal to (\ref{in:7}).
For this purpose, we show that the solution to the Lax pair,  $\Phi$,  which appears in (\ref{in:15}) 
can be expressed in terms of the operator $\B_s$.
Let us first introduce some notations.  
Let $\B_s$ be as defined in (\ref{in:4}).  Then $\A_s:=\B_s^2$ is the Airy operator acting on $L^2[0,\infty)$, which has kernel
\begin{equation}\label{pr:1}
\begin{split}
	\A_s(x,y)
	&=\int_0^\infty \Ai(x+s+\xi) \Ai (y+s+\xi)\,d\xi \\
	&= \frac{\Ai(x+s)\Ai'(y+s)-\Ai'(x+s)\Ai(y+s)}{x-y} \,.
\end{split}
\end{equation}
Define the functions $Q$ and $R$ as
\begin{equation}\label{pr:2}
	Q:= (\1 - \A_s)^{-1} \B_s \de_0\,, \qquad R:= (\1 - \A_s)^{-1} \A_s \de_0\,,
\end{equation}
where $\de_0$ is the Dirac delta function at zero.\footnote{Throughout this paper, we use the convention that 
the Dirac delta function satisfies $\int_{[0, \infty)} \delta_0(x) f(x)dx =f(0)$ for functions $f$ which are right-continuous at $0$.} 
Introduce also the functions
\begin{equation}\label{pr:3}
	\Th_1(x) :=\cos\left( \frac43 \zeta^3+(s+2x)\zeta\right)\,, \qquad 
	\Th_2(x) :=-\sin\left( \frac43 \zeta^3+(s+2x)\zeta\right)\,.
\end{equation}

The starting point of our analysis is the following explicit formula for the solution $\Phi$.
\begin{prop}\label{lax_pair}
The particular solution 
\begin{equation}\label{pr:4}
	\Phi(\z,s) = \begin{pmatrix} \Phi_1(\z, s) \\ \Phi_2(\z, s) \end{pmatrix}
\end{equation}
to the Lax pair (\ref{in:8}) satisfying the conditions (\ref{in:13}) is given by
\begin{equation}\label{pr:5}
	\Phi_{1}(\z,s) = \Th_{1}(0)+\langle \Th_{1}, R - Q \rangle_0\,,
	\qquad 
	\Phi_{2}(\z,s) = \Th_{2}(0)+\langle \Th_{2}, R + Q \rangle_0\,,
\end{equation}
where the functions $Q$, and $R$ are defined in (\ref{pr:2}), and $\langle \cdot, \cdot \rangle_0$ is the inner product on $L^2[0,\infty)$.
\end{prop}

The proof of this proposition is given in section~\ref{lax_proof}.

\bigskip

Using the definition of $R$ and $Q$, $\Phi_2(\z, s)$ in (\ref{pr:5}) can be written as 
\begin{equation}\label{pr:9-1}
\begin{aligned}
	\Phi_{2}(\z,s) = 
	\left\langle  \Th_2, \de_0+R+Q\right\rangle_0 
	&=	\left\langle \Th_2,  (\1+(\1-\B_s^2)^{-1}\B_s+(\1-\B_s^2)^{-1}\B_s^2)\de_0\right\rangle_0 \\
	&=	\left\langle \Th_2,  (\1-\B_s)^{-1}\de_0\right\rangle_0\,.
\end{aligned}
\end{equation}
Note that  the only term which depends on $\z$ is $\Th_2$. 
The function $h$ in (\ref{in:15}) is defined as an integral of $\z e^{-w\z^2}\Phi_2(\z,s)$ with respect to $\z$.  
With (\ref{in:15}) in mind, we therefore compute
\begin{equation}\label{pr:6}
\begin{aligned}
	\int_0^\infty & \z e^{-w\z^2} \sin \left(\frac43 \z^3+s\z +2x \z  \right)\, d\z \\
	&=\frac{1}{2i} \int_{-\infty}^{\infty} \z  \exp\left[i\left(\frac43 \z^3+i w\z^2+s\z +2x \z  \right)\right]\, d\z \\
&=\frac{1}{2i} \int_{-\infty}^\infty \z \exp\left[i\left(\frac13 \left(2^{2/3}\z+\frac{iw}{2^{4/3}}\right)^3+\z\left(\frac{w^2}{4} +s+2x\right)+\frac{iw^3}{48}\right)\right]\,d\z\,.
\end{aligned}
\end{equation}
Making the change of variables $\eta=2^{2/3}\z+iw2^{-4/3}$, (\ref{pr:6}) becomes
\begin{equation}\label{pr:7}
\begin{aligned}
	\frac{1}{i2^{7/3}}&e^{\frac{w^3}{24}+\frac{w}{4}(s+2x)}
	\int_{-\infty}^\infty \left(\eta -\frac{iw}{2^{4/3}}\right) \exp\left[i\left(\frac{1}{3}\eta^3+\eta\left(\frac{w^2}{2^{8/3}}+\frac{s+2x}{2^{2/3}}\right)\right)\right]\, d\eta  \\
	&= -\frac{\pi}{2^{4/3}} e^{\frac{w^3}{24}+\frac{w}{4}(s+2x)}
	\left[\frac{w}{2^{4/3}} \Ai\left(\frac{w^2}{2^{8/3}}+\frac{s+2x}{2^{2/3}}\right)
	+\Ai'\left(\frac{w^2}{2^{8/3}}+\frac{s+2x}{2^{2/3}}\right) \right]\,,
\end{aligned}
\end{equation}
since
\begin{equation}\label{pr:7-1}
\begin{aligned}
	\Ai(x)= \frac1{2\pi}
	\int_{-\infty}^\infty  \exp\left[i\left(\frac{1}{3}\eta^3+x\eta\right)\right]\, d\eta \,. 
\end{aligned}
\end{equation}
Comparing with the function $\Qua(x; t,m)$ in (\ref{in:6}),  we thus obtain
\begin{equation}\label{pr:8}
	\int_0^\infty \z e^{-w\z^2} \sin \left(\frac43 \z^3+s\z +2\z x \right)\, d\z 
=-\frac{\pi}{2^{7/3}} e^{\frac{w^3}{24}+\frac{ws}{4}} \Qua(2^{1/3} x;\,2^{-4/3} w, \,2^{-2/3} s)\,.
\end{equation}
Combining (\ref{pr:8}) with (\ref{pr:9-1}), we can write the function $h$ as
\begin{equation}\label{pr:9}
\begin{aligned}
	h(s,w)
	&=\frac{\pi}{2^{7/3}} e^{\frac{w^3}{24} + \frac{w s}{4}} 
	\left\langle \Qua(2^{1/3}x;\,2^{-4/3}w, 2^{-2/3}s), (\1-\B_s)^{-1} \de_0\right\rangle_0\,,
\end{aligned}
\end{equation}
where $x$ is the variable of integration in the inner product.

We now evaluate the integral 
\begin{equation}\label{pr:10}
	\int_s^\infty h(u,w) h(u, -w) du = 2\int_0^\infty h(2y+s, w) h(2y+s, -w) dy,
\end{equation}
in the definition of $P(s,w)$ in (\ref{in:14-1}). 
For this purpose, 
it is more convenient to work in the space $L^2(\R)$ instead of $L^2[0,\infty)$.  
Let us denote by $\tilde{\bf B}_s$ the operator which has the same kernel as (\ref{in:4}) and acts on $L^2(\R)$.  
Let $\Pi_r$ denote the projection onto $L^2[r,\infty)$
and let $\langle \cdot,\cdot \rangle$ be the inner product on $L^2(\R)$.
Then from (\ref{pr:9}), 
\begin{equation}\label{pr:11}
\begin{split}
	&\frac{2^{7/3}}{\pi}  e^{-\frac{w^3}{24} -  \frac{w (2y+s)}{4}}   h(2y+s,w) \\
	&= 
	\left\langle \Qua(2^{1/3}x;\,2^{-4/3}w, 2^{-2/3}(2y+s)),   \,\ \Pi_0 (\id-\Pi_0 \tilde{\bf B}_{2y+s}\Pi_0)^{-1} \Pi_0\delta_0\right\rangle\,,
\end{split}
\end{equation}
where once again $x$ is the variable of integration in the inner product.

We can push all dependence on $y$ to the right side of the inner product as follows.
Let $\T_r$ be the translation operator $(\T_r f)(x)= f(x+r)$. 
We clearly have $\Pi_r= \T_{-r}\Pi_0\T_r$ and $\Pi_0=\T_r\Pi_r \T_{-r}$.
Since 
$(\T_{-r}K\T_{r})(u,v)=K(u-r, v-r)$ for any kernel $K$, 
we find from the definition of $\tilde{\bf B}_s$ 
that 
$(\T_{-y}\tilde{\bf B}_{2y+s} \T_y)(u,v)= \Ai(u+v+s)=\tilde{\bf B}_s(u,v)$.
Thus, 
\begin{equation}\label{pr:12}
\begin{split}
	\Pi_0 (\id-\Pi_0 \tilde{\bf B}_{2y+s}\Pi_0)^{-1} \Pi_0
	&= \Pi_0  (\id - \T_y\Pi_y \T_{-y} \tilde{\bf B}_{2y+s} \T_y\Pi_y \T_{-y})^{-1} \Pi_0 \\
	&= \Pi_0  (\id - \T_y\Pi_y \tilde{\bf B}_s\Pi_y \T_{-y})^{-1} \Pi_0 \\
	&= \Pi_0  \T_y (\id - \Pi_y \tilde{\bf B}_s\Pi_y )^{-1} \T_{-y}\Pi_0 \\
	&= \T_y \Pi_y (\id - \Pi_y \tilde{\bf B}_s\Pi_y )^{-1} \Pi_y \T_{-y}\,.
\end{split}
\end{equation}
Notice also that the transpose of $\T_y$ is $\T_{-y}$ and that $\T_{-y}\delta_0(u)=\delta_y(u)$.  
Hence  (\ref{pr:11}) equals 
\begin{equation}\label{pr:13}
\begin{split}
	\left\langle \T_{-y} \Qua(2^{1/3}x;\,2^{-4/3}w, 2^{-2/3}(2y+s)),   \,\ 
	\Pi_y (\1 - \Pi_y \tilde{\bf B}_s\Pi_y )^{-1} \Pi_y \delta_y\right\rangle.
\end{split}
\end{equation}
From the definition of $\Qua$, 
it is easy to check that 
\begin{equation}\label{pr:14}
 	\Qua(a-r;\,t,m+r)= e^{-rt} \Qua(a;\,t,m).
 \end{equation}
Therefore (\ref{pr:13}) equals
\begin{equation}\label{pr:15}
\begin{aligned}
	&e^{-\frac12 wy}\left\langle \Qua(2^{1/3}x;\, 2^{-4/3}w, 2^{-2/3}s),   \,\ 
	\Pi_y (\id - \Pi_y \tilde{\bf B}_s\Pi_y )^{-1} \Pi_y \delta_y\right\rangle.
\end{aligned}
\end{equation}
Define the resolvent 
\begin{equation}\label{pr:15-1}
\begin{aligned}
	\bfR_y:= 
	(\id - \Pi_y \tilde{\bf B}_s\Pi_y )^{-1} -\1.
\end{aligned}
\end{equation}
Note that $\bfR_y= \Pi_y \bfR_y\Pi_y$ and its kernel $\bfR_y(x_1, x_2)$ is smooth
in $x_1, x_2\ge y$. 
Using this notation, we obtain
\begin{equation}\label{pr:15-2}
\begin{aligned}
	&\frac{2^{7/3}}{\pi}  e^{-\frac{w^3}{24} -  \frac{w s}{4}}   h(2y+s,w) 
	=\QQua_+(y) 
	+ \int_{-\infty}^\infty  \QQua_+(x)  \bfR_y(x,y) dx \;,
\end{aligned}
\end{equation}
where
we set
\begin{equation}\label{pr:15-3}
\begin{aligned}
	&\QQua_{\pm}(x):= \Qua(2^{1/3}x;\, \pm 2^{-4/3}w, 2^{-2/3}s) .
\end{aligned}
\end{equation}
Note that the integrand in (\ref{pr:15-2}) vanishes for $x<y$ due to the resolvent kernel. 

Inserting (\ref{pr:15-2}) into equation ~\eqref{pr:10}, we find, using the fact that the kernel of $\bfR_y$ is symmetric, that
\begin{equation}\label{pr:16}
\begin{aligned}
	\frac{2^{11/3}}{\pi^2} & \int_s^\infty  h(u,w) h(u, -w) du 
	= \int_0^\infty \QQua_+(y)\QQua_-(y) dy  \\
	&+  \int_0^\infty dy \int_{-\infty}^\infty  dx_2 \QQua_+(y)\bfR_y(y, x_2)\QQua_-(x_2) \\
	& + \int_0^\infty dy \int_{-\infty}^\infty  dx_1 \QQua_+(x_1)\bfR_y(x_1, y)\QQua_-(y) \\
	&+  \int_0^\infty dy \int_{-\infty}^\infty  dx_1 \int_{-\infty}^\infty  dx_2  \QQua_+(x_1)\bfR_y(x_1, y) \bfR_y(y, x_2)\QQua_-(x_2) .
	\end{aligned}
\end{equation}
Setting $m=2^{-2/3}s$ and $t=2^{-4/3}w$ in (\ref{in:14}) 
(see the relation between (\ref{in:14})  and (\ref{in:14-1})) and noting that $\rho_s(x,y)=(\1+\bfR_0)(x,y)$, 
we see that the formula (\ref{in:14}) implies the formula (\ref{in:7}) if we show that (\ref{pr:16}) is same as  
\begin{equation}\label{pr:16-1}
\begin{aligned}
	 \int_0^\infty \QQua_+(y)\QQua_-(y) dy  
	 + \int_{-\infty}^\infty  dx_1 \int_{-\infty}^\infty  dx_2  \QQua_+(x_1)\bfR_0(x_1, x_2)\QQua_-(x_2) .
\end{aligned}
\end{equation}
Now setting $y=x_1$ in the first double integral of (\ref{pr:16}) and setting $y=x_2$ in the second double integral, we see that the equality follows if we show that
\begin{equation}\label{pr:16-2}
\begin{aligned}
	 \bfR_{x_1}(x_1, x_2)+ \bfR_{x_2}(x_1, x_2)
	 + \int_0^\infty \bfR_y(x_1, y) \bfR_y(y, x_2)  dy
	= \bfR_0(x_1, x_2)
\end{aligned}
\end{equation}
for all $x_1, x_2\ge 0$. 
This is a general identity, given in the following lemma. 

\begin{lemma}\label{lem1}
Let $\K$ be an integral operator on $\R$ such that for all $y\ge 0$, 
$\Pi_y \K \Pi_y$ is bounded in  $L^2(\R)$ and 
$\1- \Pi_y \K \Pi_y$ is  invertible. 
Set $\bfR_y:=(\1-\Pi_y \K \Pi_y)^{-1}-\1$. 
Suppose that the kernel $\bfR_y(x_1, x_2)$ is continuous in $x_1, x_2\ge y$, for all $y\ge 0$. Then 
\begin{equation}\label{eq:26}
	 \bfR_{x_1}(x_1, x_2)+ \bfR_{x_2}(x_1, x_2)
	 + \int_0^\infty \bfR_y(x_1, y) \bfR_y(y, x_2)  dy
	= \bfR_0(x_1, x_2)
\end{equation}
for all $x_1, x_2\ge 0$.
\end{lemma}

The proof of this Lemma is given in Section~\ref{lem_proof}.
Therefore, we find that the formula (\ref{in:14})  is equivalent to the formula (\ref{in:7}). 
Theorem \ref{thm} is  proved.

\section{Proof of Proposition \ref{lax_pair}}\label{lax_proof}

A formula for the solution of a different Lax pair for the Hastings-McLeod solution to Painlev\'{e} II
in terms of the Airy function and $(\1-\A_s)^{-1}$ was obtained in \cite{baik_2005}
(see  (1.15) and the Remark after Lemma 1.4). 
This was obtained in a indirect way. 
A direct proof of the same formula in the spirit of \cite{TW} was obtained subsequently by Harold Widom (see the Remark after Lemma 1.4. of \cite{baik_2005}) but this proof was not published anywhere. 
It is possible to prove Proposition \ref{lax_pair} using this formula after some calculations.
Instead of following this route, we give a direct and self-contained proof of Proposition \ref{lax_pair} in this section for the benefit of the reader.
Our proof is similar to the calculation of Widom mentioned above.

\subsection{Preliminary work}
In order to prove Proposition \ref{lax_pair}, we first set up some notations and give some preliminary results which will be useful.  We follow the notations of the work \cite{TW}.\footnote{In \cite{TW}, the Airy kernel was defined without the parameter $s$ and the associated operators were defined in $[s,\infty)$. In this paper, we use the different convention that the Airy kernel contains the parameter $s$ and the associated operators are defined in $[0,\infty)$.}  Let $\D$ be the  differential operator, $(\D h)(x)= h'(x)$. 
In addition to the functions 
\begin{equation}\label{lax:6-1}
	Q:= (\1 - \A_s)^{-1} \B_s \de_0\,, \qquad R:= (\1 - \A_s)^{-1} \A_s \de_0\,,
\end{equation}
defined in (\ref{pr:2}), we introduce also the function
\begin{equation}\label{lax:5}
	P:= (\1 - \A_s)^{-1} \D \B_s \de_0\,.
\end{equation}
Let $a$ be the function
\begin{equation}\label{lax:6}
	a(x):=\Ai(x+s)=\B_s \de_0\,,
\end{equation}
and introduce also the notations
\begin{equation}\label{lax:6-a}
	q(s):= \langle \de_0, Q \rangle_0 = Q(0) \,, \qquad  p(s):= \langle \de_0, P \rangle_0 = P(0)\,,
\end{equation}
and
\begin{equation}\label{lax:7}
	u(s):= \langle a, Q \rangle_0 \,, \qquad v(s):= \langle a, P \rangle_0=\langle Q, \D a \rangle_0\,.
\end{equation}
It is shown in \cite{TW} that $q(s)$ as defined in (\ref{lax:6-a}) is in fact the Hastings-McLeod solution to the Painlev\'{e} II equation. 
For the convenience of the reader, we include a proof of this statement at the end of this subsection.

Recall the general identities for any operators $\U$ and $\V$ such that $\U$ depends on a parameter $s$:
\begin{eqnarray}\label{lax:8}
	\frac{\d }{\d s} (\1-\U)^{-1}&=& (\1-\U)^{-1} \left(\frac{\d}{\d s} \U \right) (\1-\U )^{-1}\,, \\
\label{lax:9}
	\V(\1-\U)^{-1} &=& (\1-\U)^{-1} [\V, \U](\1-\U)^{-1}+(\1-\U)^{-1}\V \, .
\end{eqnarray}
We will use these identities as well as the following identities for the operators $\A_s$ and $\B_s$, which are easy to check from their definitions (\ref{in:4}) and (\ref{pr:1}):\footnote{Here for two functions $f$ and $g$, the notation $f\otimes g$ stands for the rank one operator defined by $(f\otimes g)h:= \langle g,h\rangle f$.}
\begin{eqnarray}
\label{lax:10}
	\frac{\d}{\d s} \A_s &=&-a \otimes a \,,\\
\label{lax:11}
	\frac{\d}{\d s} \B_s &=& \D\B_s \,, \\
\label{lax:12}
	[\D, \A_s]  &=&  -a\otimes a + \A\de_0 \otimes \de_0 \, ,\\
\label{lax:13}
	[ \M,\A_s] &=&  a\otimes \D a - \D a \otimes a \, ,
\end{eqnarray}
where we have denoted by $\M$ the multiplication by $x$, $\M f(x)=xf(x)$. 

The key differential identities that we use to prove Proposition \ref{lax_pair} are summarized in the following lemma.
\begin{lem}\label{lem_diffeq}  We have the following identities:
\begin{eqnarray}
\label{lax:14}	
	\frac{\partial Q}{\partial s} &=& -u Q +P, \\
\label{lax:15}	
	\D Q &=& -u Q +P + q\RR , \\
\label{lax:16}
	\frac{\partial P}{\partial s} &=& (-2v + s+ \M)Q + uP,\\
\label{lax:17}	
	\D P &=& (-2v + s+ \M)Q + uP  + p\RR , \\
\label{lax:18}	
	\frac{\partial \RR}{\partial s} &=& -qQ,
\end{eqnarray}
and
\begin{eqnarray}
\label{lax:19}
	2v-u^2 &=&-q^2,\\
\label{lax:20}	
	q' &=& -uq+p.
\end{eqnarray}
Moreover, 
\begin{equation}
\label{lax:20+1}	
	\M \RR = pQ-qP. 
\end{equation}
\end{lem}

\begin{proof}
Most of the equations (\ref{lax:14})-(\ref{lax:18}) are straightforward to check using the identities (\ref{lax:8})-(\ref{lax:13}), and in fact most of them appear in \cite{TW}.  As an illustration, let us show how to prove (\ref{lax:16}).  We have 
\begin{equation}\label{lax:21}
\begin{aligned}
\frac{\d P}{\d s} &= \left(\frac{\d}{\d s} (\1-\A_s)^{-1}\right) \D a + (\1-\A_s)^{-1} \left(\frac{\d}{\d s} \D a\right)  \\
&=-(\1-\A_s)^{-1} a \otimes a (\1-\A_s)^{-1} \D a+(\1-\A_s)^{-1} \left(\frac{\d}{\d s}  \D a \right) \,,
\end{aligned}
\end{equation}
where we have used (\ref{lax:8}) and (\ref{lax:10}). 
Since in general $(f\otimes g)h = f \langle g,h \rangle$, 
the first term is $-Q\langle a,P\rangle = -vQ$. 
Notice that $\left( \frac{\d}{\d s} \D a \right) (x)= \Ai''(x+s)=(x+s) \Ai(x+s)= \left( (\M+s)a \right)(x)$. 
It follows that (\ref{lax:21}) is
\begin{equation}\label{lax:22}
\begin{aligned}
\frac{\d P}{\d s} &=-vQ +(\1-\A_s)^{-1}(\M+s) a \\
&=-vQ +(\M+s)(\1-\A_s)^{-1} a+[(\1-\A_s)^{-1},\M]a \\
&=-vQ+(\M+s)Q -(\1-\A_s)^{-1}[\M, \A_s](\1-\A_s)^{-1}a  \,,
\end{aligned}
\end{equation}
where we have used (\ref{lax:9}).  We now apply (\ref{lax:13}) to obtain
\begin{equation}\label{lax:23}
\begin{aligned}
\frac{\d P}{\d s} &=-vQ+(\M+s)Q -(\1-\A_s)^{-1}a\otimes \D a(\1-\A_s)^{-1}a \\
&\hspace{4cm} +(\1-\A_s)^{-1}\D a\otimes a(\1-\A_s)^{-1}a \\
&=(\M+s-v)Q-Qv+Pu\,,
\end{aligned}
\end{equation}
which is (\ref{lax:16}).

The identity (\ref{lax:19}) can be obtained as follows.  Starting with 
\begin{equation}\label{lax:24}
v= \langle Q, \D a\rangle_0,
\end{equation}
we integrate by parts to obtain (as $q=Q(0)$)
\begin{equation}\label{lax:25}
\begin{aligned}
v&= -qa(0)-\langle\D Q, a\rangle_0 \\
&=-qa(0)-\langle -uQ+P+qR, a\rangle_0 \\
&=-qa(0)-u^2-v-q\langle R, a \rangle_0\,,
\end{aligned}
\end{equation}
where we have used (\ref{lax:15}). We now evaluate $\langle R, a\rangle_0$ as
\begin{equation}\label{lax:26}
	\langle R, a \rangle_0=\left\langle (\1-\A_s)^{-1} \A_s \de_0, a \right\rangle_0
	=\left\langle -\de_0+(\1-\A_s)^{-1}  \de_0, a \right\rangle_0=-a(0)+q\,.
\end{equation}
Combining (\ref{lax:26}) and (\ref{lax:25}) gives (\ref{lax:19}).
On the other hand,  (\ref{lax:20}) is simply (\ref{lax:14}) evaluated at zero.  

Finally we prove (\ref{lax:20+1}). 
Since $\M\de_0=0$, 
\begin{equation}\label{lax:37}
\begin{aligned}
	\M \RR  = \M (\1-\A)^{-1}\A \de_0 
	= \M (-\1 + (\1-\A)^{-1})\de_0 
	= \M (\1 -\A)^{-1}\de_0. 
\end{aligned}
\end{equation}
Hence, using $\M\de_0=0$ one more time, as well as (\ref{lax:9}) and (\ref{lax:13}),
\begin{equation}\label{lax:38}
\begin{aligned}
	\M \RR 
	&= [ \M,  (\1 -\A)^{-1}] \de_0  \\
	&= (\1-\A)^{-1} [\M, \A]  (\1-\A)^{-1}\de_0 \\
	&= (\1-\A)^{-1} (a\otimes \D a - \D a \otimes a) (\1-\A)^{-1}\de_0 \\
	&= \langle \D a, (\1-\A)^{-1}\de_0\rangle Q - \langle a, (\1-\A)^{-1}\de_0\rangle_0 P \\
	&= pQ-qP. 
\end{aligned}
\end{equation}

\end{proof}

\medskip


For completeness, we give a proof that $q(s)$ solves the Painlev\'e II equation. 
In addition to Lemma \ref{lem_diffeq}, we also have 
\begin{eqnarray}
\label{lax:ex1}
	u' &=& -q^2, \\
\label{lax:ex2}
	p' &=& -2vq+sq+up \;.
\end{eqnarray}
Indeed, by differentiating the definition (\ref{lax:7}) of $u$, $u'= \langle a', Q \rangle_0
+  \langle a, \frac{\d}{\d s}Q \rangle_0$. 
By using the definition (\ref{lax:7}) of $v$ and using (\ref{lax:14}), we find that 
$u'=2v-u^2$. Now (\ref{lax:19}) implies  (\ref{lax:ex1}). 
On the other hand, (\ref{lax:ex2}) follows by evaluating (\ref{lax:17}) at zero. 
By differentiating (\ref{lax:20}) and then using (\ref{lax:ex1}) and (\ref{lax:ex2}), 
we obtain $q''=q^3+sq -uq'-2vq+up$. Inserting (\ref{lax:20}) for $q'$ and then using (\ref{lax:19}), 
this becomes $q''=2q^3+sq$. Thus $q$ solves the Painlev\'e II equation.

\subsection{Differential equations for $\Phi_{1}$ and $\Phi_{2}$}

Here we show that $\Phi_{1}$ and $\Phi_{2}$, as defined in (\ref{pr:5}), 
\begin{equation}\label{lax:3-0}
	\Phi_{1}(\z,s) = \Th_{1}(0)+\langle \Th_{1}, R - Q \rangle_0\,,
	\qquad 
	\Phi_{2}(\z,s) = \Th_{2}(0)+\langle \Th_{2}, R + Q \rangle_0\,,
\end{equation}
satisfy the 
Lax pair equations (\ref{in:8}), i.e. the differential equations
\begin{eqnarray}
\label{lax:1}
	\frac{\d}{\d\z}\Phi_1(\z,s) &=& 4\z q\Phi_1(\z,s)+(4\z^2+s+2q^2+2q')\Phi_2(\z,s)\,, \\
\label{lax:2}
	\frac{\d}{\d\z}\Phi_2(\z,s) &=& -(4\z^2+s+2q^2-2q')\Phi_1(\z,s) -4\z q\Phi_2(\z,s)\,,  
\end{eqnarray}
and
\begin{eqnarray}
\label{lax:3}
	\frac{\d}{\d s}\Phi_1(\z,s) &=& q \Phi_1(\z,s)+\z\Phi_2(\z,s)\,,  \\
\label{lax:4}
	\frac{\d}{\d s}\Phi_2(\z,s) &=& -\z \Phi_1(\z,s)-q\Phi_2(\z,s)\,.
\end{eqnarray}
We will use Lemma \ref{lem_diffeq} as well as the identities
\begin{equation}
\label{lax:27}
	\D\Th_{1}= 2\z \Th_{2} \,, \qquad 
	\D\Th_{2}=- 2\z \Th_{1} \,, 
\end{equation}
which are evident from the definition (\ref{pr:3}) of the functions.  
Using (\ref{lax:27}), integrating by parts, and applying (\ref{lax:15}), it is easy to see that
\begin{equation}\label{lax:28}
	2\z \langle \Th_{1} , Q \rangle_0 
	= \langle - \D \Th_{2} , Q \rangle_0 
	= q\Th_{2}(0) + \langle \Th_{2}, -uQ+P+qR \rangle_0\,,
\end{equation}
and
\begin{equation}\label{lax:28-1}
	2\z \langle \Th_{2} , Q \rangle_0 = - q\Th_{1}(0) - \langle \Th_{1}, -uQ+P+qR \rangle_0\,.
\end{equation}
Similarly, using (\ref{lax:17}), we see that
\begin{equation}\label{lax:28-a}
	2\z \langle \Th_{1} , P \rangle_0= p \Th_{2}(0) + \langle \Th_{2}, (-2v+s+\M) Q + uP+pR\, \rangle_0\,,
\end{equation}
and 
\begin{equation}\label{lax:28-a-1}
	2\z \langle \Th_{2} , P \rangle_0= - p \Th_{1}(0) - \langle \Th_{1}, (-2v+s+\M) Q + uP+pR\, \rangle_0\,.
\end{equation}

\medskip

We now prove equation (\ref{lax:1}).  Notice that the $\z$ dependence comes entirely from the functions $\Th_{1}$.  Differentiating (\ref{lax:3-0}) with respect to $\z$ gives
\begin{equation}\label{lax:32}
\begin{aligned}
	\frac{\d}{\d \z} \Phi_1
	&=(4\z^2+s) \Th_2(0)+ \left\langle (4\z^2+s+ 2\M) \Th_2, \RR-Q \right\rangle_0 \\
&= (4\z^2+s)\Phi_2 -2(4\z^2+s)\langle F_2, Q \rangle_0 +2\langle F_2, \M (\RR-Q)  \rangle_0\, .
\end{aligned}
\end{equation}
We are therefore reduced to showing
\begin{equation}\label{lax:33}
\begin{aligned}
	&2\z q \Th_1(0)+(q^2+q')\Th_2(0)+ 
	2\z q \langle \Th_1, \RR-Q \rangle_0 + (q^2+q') \langle \Th_2, \RR+Q \rangle_0 \\
	&\qquad  =
	 -(4\z^2+s)\langle \Th_2, Q \rangle_0 +\langle \Th_2, \M(\RR-Q) \rangle_0\, .
\end{aligned}
\end{equation}
In fact, inserting (\ref{lax:28-1}) into the term 
$-4\z^2 \langle \Th_2, Q \rangle_0 = -2\z \langle 2\z \Th_2, Q \rangle_0 $ in~\eqref{lax:33}, 
we find that it is enough to show that 
\begin{equation}\label{lax:34}
\begin{split}
	&(q^2+q')\Th_2(0) - 2\z q \langle \Th_1, Q \rangle_0 + (q^2+q') \langle \Th_2, \RR+Q \rangle_0 \\
	&\qquad= -2\z u \langle \Th_1, Q \rangle_0 + 2\z \langle \Th_1, P \rangle_0
	 -s\langle \Th_2, Q \rangle_0 +\langle \Th_2, \M(\RR-Q) \rangle_0\, .
\end{split}
\end{equation}
We now apply (\ref{lax:28}) to the term $-2\z q \langle \Th_1, Q \rangle_0 $ and (\ref{lax:28-a-1}) to the term $2\z\langle \Th_2, P \rangle_0$ in (\ref{lax:34}) 
and find that it is enough to show that 
\begin{equation}\label{lax:35}
\begin{split}
	&(q'+uq-p)\Th_2(0)+(q'+2v+uq-u^2+q^2)\langle \Th_2, Q \rangle_0 \\
	&\qquad\qquad  +(q'+uq-p)\langle \Th_2, R\rangle_0 - q \langle \Th_2, P \rangle_0 = \langle \Th_2, \M R \rangle_0\,.
\end{split}
\end{equation}
Using~\eqref{lax:19} and~\eqref{lax:20}, this is equivalent to 
\begin{equation}\label{lax:36}
	p  \langle \Th_2, Q \rangle_0  - q \langle \Th_2, P  \rangle_0
	= \langle \Th_2, \M \RR \rangle_0\, .
\end{equation}
This follows from (\ref{lax:20+1}), 
which proves (\ref{lax:1}).  The proof of (\ref{lax:2}) is similar.

\medskip

We now prove equation (\ref{lax:3}).  Differentiating (\ref{pr:5}) with respect to $s$ gives
\begin{equation}\label{lax:29}
\begin{aligned}
	\frac{\d}{\d s} \Phi_1 
	&= \z \Th_2(0) + \langle \z \Th_2, \RR-Q  \rangle_0 
	+  \left\langle \Th_1, \frac{\d}{\d s} (\RR-Q) \right\rangle_0 \\
	&= \z \Phi_2 -2\z \langle \Th_2, Q \rangle_0
	+  \left\langle F_1, \frac{\d}{\d s} ( \RR-Q)  \right\rangle_0.
\end{aligned}
\end{equation}
Using (\ref{lax:28-1}), this is
\begin{equation}\label{lax:30}
\begin{aligned}
	\frac{\d}{\d s} \Phi_1 
	=  \z \Phi_2 + q\Th_1(0)  + \langle \Th_1, -uQ+P+qR \rangle_0
	+  \left\langle \Th_1, \frac{\d}{\d s} ( \RR-Q)  \right\rangle_0.
	\end{aligned}
\end{equation}
We now apply ~\eqref{lax:14} and ~\eqref{lax:18} to obtain
\begin{equation}\label{lax:31}
	\frac{\d}{\d s} \Phi_1 
	=  \z \Phi_2 + q\Th_1(0) +q \langle \Th_1, R-Q \rangle_0\,,
\end{equation}
which is clearly $\z \Phi_2 +q \Phi_1$.  This proves (\ref{lax:3}).  The proof of (\ref{lax:4}) is nearly identical.

\subsection{Asymptotics of $\Phi_{1}$ and $\Phi_2$}

In order to complete the proof of Proposition \ref{lax_pair}, we must show that the functions $\Phi_{1}$ 
and $\Phi_2$ as defined in (\ref{pr:5}) have the real asymptotics (\ref{in:13}). 
Since $\Th_{1}(0)$ and $\Th_2(0)$ are precisely the leading terms of the asymptotics (\ref{in:13}), 
it is enough to show that 
\begin{equation}\label{as:1}
	\langle \Th_{1}, R - Q \rangle_0 = O\left(\z^{-1}\right)\,,
	\qquad \langle \Th_{2}, R + Q \rangle_0 = O\left(\z^{-1}\right)\,,
\end{equation}
as $\z \to \pm \infty$.
Since the only dependence on $\z$ is in $\Th_{1}$ and $\Th_2$, 
from the definitions (\ref{pr:2}) of $Q$ and $P$, 
the asymptotics (\ref{as:1}) is proved if we show that 
\begin{equation}\label{as:2}
	\int_0^\infty e^{i\left(\frac43 \zeta^3+(s+2x)\zeta \right)} \Ai(x+s+\xi) dx = O\left(\z^{-1}\right)\,
\end{equation}
as $\z \to \pm \infty$ uniformly in $\xi\in [0,\infty)$ for a fixed $s\in \R$. 
But from the definition of the Airy function, the integral in (\ref{as:2}) equals 
\begin{equation}\label{as:3}
	\frac1{2\pi} \int_0^\infty e^{i\left(\frac43 \zeta^3+(s+2x)\zeta \right)} 
	\int_\Sigma e^{i\left(\frac43 \eta^3+(x+s+\xi)\eta \right)} d\eta dx \;,
\end{equation}
where $\Sigma$ can be taken to be a contour in $\C_+$ whose asymptotes are the rays of angle $\pi/6$ and $5\pi/6$.  
Changing the order of integrals and integrating in $x$, we obtain, by noting that the real part of $i(2\zeta+\eta)x<0$, that (\ref{as:3}) equals 
\begin{equation}\label{as:4}
	-\frac{e^{i\left(\frac43 \zeta^3+s\zeta \right)} }{2\pi i} 	
	\int_\Sigma \frac{e^{i\left(\frac43 \eta^3+(s+\xi)\eta \right)}}{2\zeta+\eta}  d\eta \;.
\end{equation}
This is clearly $O\left(\z^{-1}\right)$ uniformly in $\xi\in [0,\infty)$ and for a fixed $s\in\R$. 
Thus, (\ref{as:2}) is proved and this completes the proof of Proposition \ref{lax_pair}.

\section{Proof of Lemma \ref{lem1}}\label{lem_proof}

The following simple proof is due to Percy Deift. 
This simplifies the original proof of ours which was more involved.

We use the notations
\begin{equation}\label{prlem1}
	\K_y := \Pi_y \K \Pi_y\,, \qquad \bfR_y:=(\id-\K_y)^{-1}-\1 \,.
\end{equation}
We evaluate $\frac{\pd}{\pd y} \bfR_y(x_1, x_2)$ for $y>0$. 
Let $\ep>0$ and consider
\begin{equation}\label{prlem2}
\begin{aligned}
	\bfR_{y+\ep}-\bfR_y&=(\id-\K_{y+\ep})^{-1}-(\id-\K_y)^{-1} \\
&=(\id-\K_{y+\ep})^{-1}(\K_{y+\ep}-\K_{y})(\id-\K_y)^{-1} \\
	&=(\1+\bfR_{y+\ep})\big((\Pi_{y+\ep}-\Pi_y) \K \Pi_{y+\ep}+\Pi_{y} \K (\Pi_{y+\ep}-\Pi_y)\big)
(\1+\bfR_y) \,,
\end{aligned}
\end{equation}
where the last equality is obtained by adding and subtracting $\Pi_y \K \Pi_{y+\ep}$ inside the parentheses. 
Since $\bfR_{y+\ep}=  \Pi_{y+\ep}\bfR_{y+\ep}\Pi_{y+\ep}$, we have  
$\bfR_{y+\ep}(\Pi_{y+\ep}-\Pi_y)=0$.
Also note that 
\begin{equation}\label{prlem6}
	\lim_{\ep \downarrow 0}\frac{\Pi_{y+\ep}-\Pi_y}{\ep} = -\M_{\de_y}\,,
\end{equation} 
where $\M_{\de_y}$ is the operator of multiplication by $\delta_y$.  
Therefore, we obtain
\begin{equation}\label{prlem7}
	\frac{\pd}{\pd y_+} \bfR_y 
	= \big(-\M_{\de_y}\K \Pi_{y} - \Pi_{y} \K \M_{\de_y} 
	-\bfR_{y}\Pi_{y} \K \M_{\de_y}\big)(\1+\bfR_y)\,,
\end{equation}
where $\frac{\pd}{\pd y_+}$ means that this is a right-sided derivative.  
Note that $\M_{\de_y}\K= \M_{\de_y}\Pi_y\K$ and $\K \M_{\de_y}= \K\Pi_y \M_{\de_y}$. 
Hence (\ref{prlem7}) can be written as 
\begin{equation}\label{prlem7-1}
	\frac{\pd}{\pd y_+} \bfR_y 
	= -\big(\M_{\de_y}\K_{y} +\K_y \M_{\de_y}
	+\bfR_{y} \K_y \M_{\de_y}  \big)(\1+\bfR_y).
\end{equation}
From the definition (\ref{prlem1}) of $\bfR_y$ we have $\bfR_y\K_y=\K_y\bfR_y= \bfR_y-\K_y$. 
Using this, (\ref{prlem7-1}) simplifies to 
\begin{equation}\label{prlem7-2}
	\frac{\pd}{\pd y_+} \bfR_y 
	= -\M_{\de_y}\bfR_{y} - \bfR_y \M_{\de_y} - \bfR_y\M_{\de_y}\bfR_y.
\end{equation}
A similar computation for the left-sided derivative shows that (\ref{prlem7-2}) holds as a derivative from both sides. 
Hence we obtain
\begin{equation}\label{prlem7-3}
	\frac{\pd}{\pd y} \bfR_y (x_1, x_2)
	= -\de_y(x_1)\bfR_{y}(x_1, x_2) - \bfR_y(x_1, x_2)\de_y(x_2) - \bfR_y(x_1, y)\bfR_y(y, x_2).
\end{equation}
Now we integrate the both sides from $y=0$ to $y=\infty$ and find that 
 \begin{equation}\label{prlem7-4}
	- \bfR_0 (x_1, x_2)
	= -\bfR_{x_1}(x_1, x_2) - \bfR_{x_2}(x_1, x_2) - \int_0^\infty \bfR_y(x_1, y)\bfR_y(y, x_2) dy.
\end{equation}
This is precisely  (\ref{eq:26}).

\end{document}